\documentclass[preprint,12pt]{article}
\usepackage{amsmath,epsfig}
\usepackage{amssymb}
\usepackage{color}
\usepackage[english]{babel}
\usepackage{graphicx}
\usepackage{times}
\usepackage{verbatim}
\usepackage{multirow}

\newtheorem{thm}{Theorem}[section]
\newtheorem{lem}[thm]{Lemma}
\newtheorem{defn}[thm]{Definition}
\newtheorem{cor}[thm]{Corollary}
\newtheorem{prop}[thm]{Proposition}
\newtheorem{rem}[thm]{Remark}

\newtheorem{alg}[thm]{Algorithm}

\newenvironment{proof}{\noindent\textsc{Proof: }}{\hfill$\fbox{}$\par\medskip\par}

\newenvironment{aenum}{\begin{enumerate}
 
 }{\end{enumerate}}

\newcommand{\R}{{\mathbb R}}
\newcommand{\Z}{{\mathbb Z}}
\newcommand{\N}{{\mathbb N}}

\newcommand{\cF}{{\cal F}}

\def\mapright#1{\stackrel{#1}{\longrightarrow}}

\def\mapdown#1{\Big\downarrow\rlap{$\vcenter{\hbox{$\scriptstyle#1$}}$}}

\newcommand{\ma}{\texttt{m}\,} 

\newcommand{\sS}{\texttt{S}} 
\newcommand{\sA}{\texttt{A}} 
\newcommand{\sB}{\texttt{B}}
\newcommand{\sC}{\texttt{C}} 
\newcommand{\sD}{\texttt{D}} 

\newcommand{\mI}{\texttt{$\max\!$I}\,}

\title{Reducing Complexes in Multidimensional Persistent Homology Theory}

\author{Madjid Allili,
Tomasz Kaczynski, and Claudia Landi}

\begin{document}

\maketitle

\begin{abstract}
The Forman's discrete Morse theory appeared to be useful for providing filtration--preserving reductions of complexes
in the study of persistent homology. So far, the algorithms computing discrete Morse matchings have only been used for
one--dimensional filtrations. This paper is perhaps the first attempt in the direction of extending such algorithms to
multidimensional filtrations. Initial framework related to Morse matchings for the multidimensional setting is proposed,
and a matching algorithm given by King, Knudson, and Mramor is extended in this direction. The correctness of the
algorithm is proved, and its complexity analyzed. The algorithm is used for establishing a reduction of a simplicial
complex to a smaller but not necessarily optimal cellular complex. First experiments with filtrations of triangular
meshes are presented.
\end{abstract}


\section{Introduction}
\label{sec:intro}

The {\em persistent homology} has been intensely developed in the last decade as a tool for studying problems of two kinds.
One is the topological analysis of discrete data, e.g. {\em point-cloud data}, where the chosen framework is a discrete linear
filtration of a simplicial complexes. The first contributions in this direction given by Edelsbrunner et al. in \cite{EdLeZo02},
and later by Carlsson and Zomorodian \cite{CaZo05} opened a new direction in research. The other one is the study of shape
similarity by {\em shape-from-function} methods, where the framework is the filtration of a topological triangulable space by
the values of a continuous function called {\em measuring function}. The $0$--dimensional persistent homology case, where the
topological invariants are based on the number of connected components, where known under the name of the {\em size function}
theory since the paper by Frosini \cite{Fr91}. The applications of persistent homology to shape similarity are studied
by~\cite{VeUr*93,CaZo*05,DiLa*12}. The two frameworks, discrete and continuous, have been extended to the multiparameter
filtration case called {\em multidimensional persistence}, where the filtration is set up with respect to a parameter space
that is no longer ordered linearly \cite{CaZo07,BiCe*08,CaDiFe10,CeDi*10}. In the continuous setting this gives rise to
{\em multidimensional measuring functions}, that is functions with values in $\R^k$. In \cite{CEFKL13} the relation between
the discrete and continuous settings is established.

In parallel, another mathematical theory which became increasingly popular in computational sciences is the Forman's
{discrete Morse theory} \cite{For98}. We will not elaborate on all possible applications of this theory in visualization,
imaging, computational geometry and other fields but just point out the one to computing persistence. The effective
computation of the persistent homology is a challenge due to a huge size of complexes built from data, for instance,
via meshing techniques. The discrete Morse theory enables algorithms reducing a given complex (simplicial, cubical, or cellular)
to a much smaller cellular complex, homotopically equivalent to the initial one, by means of {\em Morse matching}, also called
{\em Morse pairing}. An ultimate goal is often to reduce the complex to an optimal one, where all remaining cells are topologically
significant. If a reduction by Morse pairings can be performed in a filtration--preserving way, that leads to a faster persistent
homology computation. This goal motivated the contributions of King, Knudson, and Mramor \cite{KinKnuMra05}, Mischaikow and Nanda
\cite{MiNa}, Robins et al.~\cite{RobWooShe11}, and D{\l}otko and Wagner \cite{DloWag}.

Given a complex and a partial pairing of its cells, the paired cells form a discrete vector field in the language
of discrete Morse theory and can be reduced in pairs so to obtain at each step a new complex homotopically equivalent
to the previous one. The final complex consists of unpaired cells that are also called critical cells. First, we give
an algorithm that constructs a Morse matching for a given complex and we prove its correctness and analyze its complexity.
Then, we go on proving that given a multifiltration on the initial complex, the reduction process yields
a new multifiltration consisting of smaller complexes and which has the same persistence homology as the initial one. As pointed
out in~\cite{MiNa} for the one dimensional case, the complexity of computing multidimensional persistence homology of a filtration
is essentially determined by the sizes of its complexes. This motivates this approach of reducing the initial complexes for achieving
a low computational cost in the persistence homology computation. Our matching algorithm can be considered as an extension to the
multidimensional setting of the algorithms given in King et al.~\cite{KinKnuMra05} and Cerri et al.~\cite{CeFrKrLa11}.
The multidimensionality is symbolized by the function defined on the vertices of the complex. The algorithm is of iterative
and recursive nature. It considers every vertex of the complex and builds a partial matching recursively on its lower link
before extending it to the entire complex. When the dimension is fixed and the number of cofaces of every cell in the complex
is bounded above by a fixed constant, we prove that the computational complexity of the algorithm is linear in the number of
vertices of the initial complex.

So far, the algorithms for discrete Morse pairings have only been used for one--parameter filtrations. There does not yet exist
a systematic extension of the Forman's discrete Morse theory to the multiparameter case, and this goal offers challenges both on
theoretical as on computational level. This paper is the first attempt in this direction.

The paper is organized as follows. In Section~\ref{prel}, we recall definitions and some known facts about $S$-complexes,
multidimensional persistent homology, acyclic matchings, and reduction of $S$--complexes.
In Section~\ref{sec:main}, we propose initial definitions of Morse pairings for the multidimensional setting and we extend
the algorithm given by King, Knudson, and Mramor. We next prove the correctness of the algorithm. Note that we do not claim
to obtain an optimal cellular complex. The set of cells we call {\em critical} is simply the set of all unpaired cells and,
typically, this is not an optimal complex. We next establish our filtration--preserving complex reduction method.
In Section~\ref{sec:experiments}, we present our first experiments with multifiltrations of triangular meshes. These
experiments show a fair rate of reduction but not an optimal one in the sense that the remaining cells are not all
relevant in the computation of persistence homology. An improvement of our methods towards the optimality is a research in progress.

\section{Preliminaries}
\label{prel}

\subsection{$\sS$-complexes}
\label{sec:s-complexes}

We shall use the combinatorial framework of $\sS$-complexes introduced in \cite{MrBa09}. Let $R$ be a principal ideal domain
(PID) whose invertible elements we call {\em units}. Given a finite set $X$, let $R(X)$ denote the free module over $R$ generated by $X$.

Let $\sS$ be a finite set with a gradation $\sS_q$ such that $\sS_q =\emptyset$ for $q < 0$.
Then $R(\sS_q)$ is a gradation of $R(\sS)$ in the category of moduli over the
ring $R$. For every element $\sigma \in \sS$ there exists a unique number $q$ such that $\sigma \in \sS_q$. This number will
be referred to as the dimension of $\sigma$ and denoted $\dim \sigma$.

Let $\kappa : \sS \times \sS \to R$ be a function such that, if $\kappa(\sigma, \tau) \ne 0$, then $\dim \sigma = \dim \tau + 1$.

We say that $(\sS, \kappa)$ is an {\em $\sS$-complex} if $(C_*(\sS),\partial^\kappa_*)$ with
$C_q(\sS):=R(\sS_q)$ and $\partial^\kappa_q : C_q(\sS) \to C_{q-1}(\sS)$ defined on generators $\sigma\in \sS$ by
$$\partial^\kappa(\sigma) :=\sum_{\tau\in \sS}\kappa(\sigma, \tau)\tau$$
is a free chain complex with base $\sS$. The map $\kappa$ will be referred to as
the {\em coincidence index}. If $\kappa(\sigma, \tau) \ne 0$, then we say that $\tau$ is a {\em primary face} of $\sigma$
and $\sigma$ is a {\em primary coface} of $\tau$.
We say that $\tau$ is a {\em face} of $\sigma$ and $\sigma$ is a {\em coface} of $\tau$ if there is a sequence of
generators ordered by the primary face relation starting with $\tau$ and ending with $\sigma$.

By the homology of an $\sS$-complex $(\sS, \kappa)$ we mean the homology of the
chain complex $(C_*(\sS),\partial^\kappa_* )$, and we denote it by $H_*(\sS, \kappa)$ or simply by $H_*(\sS)$.

The choice of $R$ a PID is made for the sake of homology computations, and also because in the proof of Proposition \ref{incidence}
we actually use the cancellation law.\\

A special case of an $\sS$ complex is the simplicial complex.
A $q$-simplex $\sigma = [v_0,v_1, \ldots, v_q]$ in $\R^d$ is the convex hull of
$q + 1$ affinely independent points $v_0$,$v_1$, $\ldots$, $v_q$ in $\R^d$, called the vertices
of $\sigma$. The number $q$ is the dimension of the simplex. A face of $\sigma$
is a simplex whose vertices constitute a subset of $(v_0,v_1,\ldots ,v_q)$. A
simplicial complex consists of a collection $\sS$ of simplices such that every
face of a simplex in $\sS$ is in $\sS$, and the intersection of two simplices
in $\sS$ is their common face. The simplicial complex $\sS$ has a natural
gradation $({\sS}_q)$, where $\sS_q$ consists of simplices of dimension $q$. Since a
zero dimensional simplex is the singleton of its unique vertex, ${\sS}_0$ may
be identified with the collection of all vertices of all simplices in the
simplicial complex ${\sS}$.

Assume an ordering of ${\sS}_0$ is given and every simplex $\sigma$ in ${\sS}$ is coded
as $[v_0,v_1,\ldots v_q]$, where the vertices $v_0,v_1,\ldots v_q$ are listed according
to the prescribed ordering of ${\sS}_0$. By putting
$$\kappa(\sigma, \tau) :=\left\{\begin{array}{ll}
(-1)^i & \mbox{if  $\sigma= [v_0,v_1, \ldots ,v_q]$}\\
 & \mbox{and  $\tau= [v_0,v_1, \ldots,v_{i-1},v_{i+1},\ldots ,v_q]$}\\
0 & \mbox{otherwise.}\end{array}\right.$$
we obtain an $\sS$-complex whose chain complex is the classical simplicial
chain complex used in simplicial homology.

\subsection{Multidimensional Persistent Homology}
\label{sec:md-pers-hom}
Let $(\sS, \kappa)$ be an $\sS$-complex. A {\em multi-filtration} of $\sS$ is a family $\cF=\{\sS^\alpha\}_{\alpha\in \R^k}$ of
subsets of $\sS$ with the following properties:
\begin{aenum}
\item $\cF$ is nested with respect to inclusions, that is $\sS^\alpha\subseteq \sS^{\beta}$, for every $\alpha\preceq \beta$,
where $\alpha\preceq \beta$ if and only if $\alpha_i \leq \beta_i$ for all $i=1,2,\ldots, k$;
\item $\cF$ is non-increasing on faces, that is, if $\sigma \in \sS^\alpha$ and $\tau$ is a face of $\sigma$ then $\tau \in \sS^\alpha$.
\end{aenum}

Persistence is based on analyzing the homological changes occurring along the filtration as $\alpha$ varies. This analysis is carried out
by considering, for $\alpha\preceq\beta$, the homomorphism
\[
H_*(j^{(\alpha,\beta)}): H_*(\sS^{\alpha}) \to
H_*(\sS^{\beta}).
\]
induced by the inclusion map $j^{(\alpha,\beta)}:\sS^{\alpha}\hookrightarrow \sS^{\beta}$.

The image of the map $H_q(j^{(\alpha,\beta)})$ is  known as the {\em $q$'th persistent homology group} of the filtration at $(\alpha,\beta)$
and we denote it by $H_q^{\alpha,\beta}(\sS)$. It contains the homology classes of order $q$ born not later than $\alpha$ and still alive at $\beta$.

The framework described so far for general filtrations can be specialized in various directions. A case relevant for a simplicial complex
is when the filtration is induced by the values of a function defined at its vertices. Let ${\sS}$ be a simplicial complex. Given a function
$f:{\sS}_0\to \R^k$, it induces on ${\sS}$ the so-called {\em sublevel set filtration}, defined as follows:
\[
{\sS}^\alpha=  \{\sigma=[v_0,v_1,\ldots,v_q]\in {\sS} \mid f(v_i)\preceq \alpha, \ i=0,\ldots ,q\}.
\]
We will call the function $f$ a {\em measuring function}.

\subsection{Acyclic Partial Matchings}
\label{sec:acyclic-match}

Let $(\sS, \kappa)$ be an $\sS$-complex. A {\em partial matching} $(\sA,\sB,\sC,\ma)$ on $(\sS, \kappa)$ is a partition of $\sS$ into three
sets $\sA,\sB,\sC$ together with a bijective  map $\ma: \sA\to \sB$  such that, for each $\tau\in \sB$, $\kappa(\ma(\tau),\tau)$ is invertible.
Observe that, in particular, $\ma(\tau)$ is a primary coface of $\tau$.

A partial matching $(\sA,\sB,\sC,\ma)$ on $(\sS, \kappa)$ is called {\em acyclic} if there does not exist a sequence
\begin{equation}\label{loop}
\sigma_0,\tau_0,\sigma_1,\tau_1,\ldots ,\sigma_p,\tau_p,\sigma_{p+1}
\end{equation}
such that, $\sigma_{p+1}=\sigma_0$, and, for each $i=0,\ldots ,p$, $\sigma_{i+1}\ne \sigma_i$, $\tau_i=\ma(\sigma_i)$, and $\tau_i$ is a primary
coface of  $\sigma_{i+1}$.

A convenient way to reformulate the definition of an acyclic partial matching is via Hasse diagrams. The {\em Hasse diagram} of $(\sS, \kappa)$
is the directed graph whose vertices
are elements of $\sS$, and the edges are given by primary face relations and oriented from the larger element to the smaller one. Given a partial
matching $(\sA,\sB,\sC,\ma)$ on  $(\sS, \kappa)$,
we change the orientation of the edge $(\tau, \sigma)$ whenever $\tau = \ma(\sigma)$. The acyclicity condition says that the oriented graph
obtained in this way, which is also called the \emph{modified Hasse diagram} of $(\sS, \kappa)$, has no nontrivial cycles.
A directed graph with no directed cycles is called a directed acyclic graph (DAG). Thus, a partial matching
$(\sA,\sB,\sC,\ma)$ on $(\sS, \kappa)$ is {\em acyclic} if its corresponding \emph{modified Hasse diagram} is a DAG.

\subsection{Reductions}
\label{sec:reductions}

We describe here a reduction construction which was introduced in \cite{KaMrSl98} for finitely generated chain complexes, also presented in
\cite[Chapter 4]{KacMisMro04}. The construction was reused in \cite{MrBa09} for the purposes of the coreduction method and, recently, in
\cite{MiNa} for the one-dimensional filtration of $S$-complexes, which is perhaps the closest reference for the purposes of this paper.

Let  $(\sA,\sB,\sC,\ma)$ be a partial matching (not necessarily acyclic) on an $\sS$-complex $(\sS, \kappa)$. Given $\sigma\in \sA$, a new
$\sS$-complex $(\overline{\sS}, \overline{\kappa})$ is constructed by setting $\overline{\sS}=\sS\setminus \{\ma(\sigma),\sigma\}$, and
$\overline{\kappa}:\overline{\sS}\times \overline{\sS}\to R$,
\begin{equation}\label{eq:kappa-bar}
\overline{\kappa}(\eta,\xi)=\kappa(\eta,\xi)-\frac{\kappa(\eta,\sigma)\kappa(\ma(\sigma),\xi)}{\kappa(\ma(\sigma),\sigma)}.
\end{equation}
Note that $\kappa(\ma(\sigma),\sigma)$ is invertible by the definition of a partial matching. We say that $(\overline{\sS}, \overline{\kappa})$
is obtained from $(\sS,\kappa)$ by a {\em reduction} of the pair $(\ma(\sigma),\sigma)$.

A pair of linear maps $\pi: C_*(\sS)\to C_*(\overline{\sS})$ and $\iota: C_*(\overline{\sS})\to C_*(\sS)$ is defined on generators by setting
\begin{equation}\label{eq:pi}
\pi(\tau)=\left\{\begin{array}{ll}
0 & \mbox{if $\tau=\ma(\sigma)$}\\
-\sum_{\xi\in \overline{\sS}}\frac{\kappa(\ma(\sigma),\xi)}{\kappa(\ma(\sigma),\sigma)}\xi & \mbox{if $\tau=\sigma$}\\
\tau & \mbox{otherwise}
\end{array}\right.
\end{equation}
and
\begin{equation}\label{eq:iota}
\iota(\tau)=\tau-\frac{\kappa(\tau,\sigma)}{\kappa(\ma(\sigma),\sigma)}\ma(\sigma).
\end{equation}

It is well known \cite{KaMrSl98} that $C_*(\overline{\sS})$ is a well-defined chain complex, and that $\pi$ and $\iota$ are chain equivalences with the
chain homotopy $D_*: C_*(\sS)\to C_{*+1}(\sS)$ given on generators $\tau \in \sS_q$, $q\in \Z$, by
\begin{equation}\label{eq:chain-homotopy}
D_q(\tau)=\left\{\begin{array}{ll}
\frac{1}{{\kappa(\ma(\sigma),\sigma)}}\ma(\sigma) & \mbox{if $\tau=\sigma$}\\
0 & \mbox{otherwise}
\end{array}\right.
\end{equation}

As a consequence, $H_*(\sS)=H_*(\overline{\sS})$.

Let  $(\sA,\sB,\sC,\ma)$ be an acyclic partial matching  on an $\sS$-complex $(\sS, \kappa)$. Let $(\overline{\sS}, \overline{\kappa})$ be obtained
from $(\sS,\kappa)$ by reduction of the pair $(\ma(\sigma),\sigma)$, $\sigma\in \sA$.

\begin{prop}\label{incidence}
If $(\sA,\sB,\sC,\ma)$ is acyclic then, for any $\tau\in \sA\setminus \{\sigma\}$, $\overline{\kappa}(\ma(\tau),\tau)$ is invertible. Furthermore,
$\overline{\kappa}(\ma(\tau),\tau)=\kappa(\ma(\tau),\tau)$.
\end{prop}

\begin{proof}
By definition,
\[
\overline{\kappa}(\ma(\tau),\tau)=\kappa(\ma(\tau),\tau)-\frac{\kappa(\ma(\tau),\sigma)\kappa(\ma(\sigma),\tau)}{\kappa(\ma(\sigma),\sigma)}.
\]
If $\kappa(\ma(\tau),\sigma)\kappa(\ma(\sigma),\tau)=0$, then $\overline{\kappa}(\ma(\tau),\tau)=\kappa(\ma(\tau),\tau)$ is invertible. Otherwise,
$\kappa(\ma(\tau),\sigma)\ne 0$ and $\kappa(\ma(\sigma),\tau)\ne 0$. Hence $\sigma$ is a primary face of $\ma(\tau)$ and $\tau$ is a primary face of
$\ma(\sigma)$. On the other hand, by definition of $\ma$,  $\sigma$ is a primary face of $\ma(\sigma)$ and $\tau$ is a primary face of $\ma(\tau)$.
But this contradicts the assumption that the partial matching is acyclic. In conclusion, necessarily $\overline{\kappa}(\ma(\tau),\tau)=\kappa(\ma(\tau),\tau)$.
\end{proof}

\begin{cor}\label{cor:reduced-matching}
Let  $(\sA,\sB,\sC,\ma)$ be an acyclic partial matching on $(\sS,\kappa)$. Given a fixed $\sigma\in \sA$,
define $\overline{\sA}=\sA\setminus \{\sigma\}$,
$\overline{\sB}=\sB\setminus \{\ma(\sigma)\}$, $\overline{\ma}=\ma_{|\overline{\sA}}$, and $\overline{\sC}=\sC$.
Then $(\overline{\sC},\overline{\ma}:\overline{\sA}\to \overline{\sB})$ is an acyclic partial matching on $(\overline{\sS},\overline{\kappa})$.
\end{cor}

\begin{proof} The bijectivity of $\overline{\ma}$ is obvious by definition. The invertibility of $\overline{\kappa}(\ma(\tau),\tau)$ has been
just proved in Proposition~\ref{incidence}. A cycle in the Hasse diagram of $(\overline{\sS},\overline{\kappa})$ is also a cycle in $(\sS,\kappa)$,
hence the acyclicity condition follows.
\end{proof}

Finally, we define the induced filtration on $\overline{\sS}$.

\begin{defn}\label{def:filtration-reduced}
Let ${\cF}=\{\sS^\alpha\}_{\alpha\in\R^k}$ be a multifiltration on $\sS$. Then $\overline{\cF}=\{\overline{\sS}^\alpha\}_{\alpha\in\R^k}$ is the multifiltration on
$\overline{\sS}$ defined by setting, for each $\tau \in \overline{\sS}$,
\[
\tau \in \overline{\sS}^\alpha \iff \tau\in \sS^\alpha.
\]
\end{defn}

\section{Main Results}
\label{sec:main}

\subsection{Matching Algorithm}
\label{sec:match-alg}
In this section we consider a finite simplicial complex $\sS$ together with a function $f : \sS_0 \to \R^k$ inducing the sublevel
set filtration ${\cF} = \{{\sS}^\alpha\}_{\alpha\in\R^k}$.
Given two values $\alpha =(\alpha_i),\beta=(\beta_i)\in\R^k$ we set $\alpha \prec \beta$ (resp. $\alpha \preceq \beta$) if and only
if $\alpha_i < \beta_i$ (resp. $\alpha_i \leq \beta_i$) for every $i$ with $1\le i\le k$. Moreover we write
$\alpha \precneqq \beta$ whenever $\alpha \preceq \beta$ and $\alpha \ne \beta$.

\subsubsection{Indexing Map for Vertices}
By definition, an indexing map on the vertices of the complex $\sS$ is any one-to-one map $I:{\sS}_0 \to \N$. Our objective is
to build an indexing map $I$ such that, for each
$v,w\in {\sS}_0$ with $v\ne w$, $f(v)\precneqq  f(w)$ implies $I(v)<I(w)$. For this purpose, we will use
topological sorting of the vertices in ${\sS}_0$.

We recall that a topological sorting of a directed graph is a linear ordering of its vertices such that for every directed edge
$(u, v)$ from vertex $u$ to vertex $v$, $u$ precedes $v$ in the ordering. This ordering is possible if and only if the graph
has no directed cycles, that is, if it is a DAG. A simple well known algorithm (see~\cite{Wikipedia14,Kahn62}) for this task consists
of successively finding vertices of the DAG that have no incoming edges and placing them in a list for the final sorting. Note that
at least one such vertex must exist in a DAG, otherwise, the graph must have at least one directed cycle. Let \texttt{L} denote the list
that will contain the sorted vertices of ${\sS}_0$ and \texttt{I} the list of vertices or nodes in the DAG with no incoming edges.
The algorithm consists of two nested loops as summarized below.

\begin{alg}{\em [Topological sorting] \label{alg:sorting}\\
\textbf{while} there are vertices remaining in I \textbf{do}\\
\hspace*{1cm}\textbf{remove} a vertex u from I\\
\hspace*{1cm}\textbf{add} u to L\\
\hspace*{1cm}\textbf{for} each vertex v with an edge e from u to v \textbf{do}\\
\hspace*{2cm}\textbf{remove} edge e from the DAG\\
\hspace*{2cm}\textbf{if} v has no other incoming edges \textbf{then}\\
\hspace*{3cm}\textbf{insert} v into I\\
\noindent{\bf End}
}
\end{alg}

When the graph is a DAG, there exists at least one solution for the sorting problem, which is not necessarily unique.
We can easily see that the algorithm visits potentially every node and every edge of the DAG, therefore its running time
is linear in the number of nodes plus the number of edges in the DAG.

\begin{lem}
\label{lem:index}
There exists an injective function $I:{\sS}_0 \to \N$ such that, for each $v,w\in  {\sS}_0$ with $v\ne w$, $f(v)\precneqq  f(w)$ implies $I(v)<I(w)$.
\end{lem}

\begin{proof}
Let us denote by $N$ the cardinality of $\sS_0$. The set ${\sS}_0$ can be represented in a directed graph where each vertex is a node,
and a directed edge is drawn between two vertices $u, w \in {\sS}_0$ if and only if $f(v)\precneqq  f(w)$. It is easily seen that we
actually obtain a directed acyclic graph (DAG), since a directed cycle in ${\sS}_0$ leads to the relation $f(u) \precneqq f(u)$
for some vertex $u \in {\sS}_0$, which is a contradiction. The topological sorting algorithm outlined above will allow to sort and store the vertices
in $\sS_0$ in an array $A$ of size $N$, with indexes that can be chosen from 1 to $N$. It follows that the map $I:{\sS}_0 \to \N$ that
associates to every vertex its index in the array $A$ is an injective map on ${\sS}_0$. Moreover, and due to topological sorting, $I$ satisfies the constraint
that for $v,w\in  {\sS}_0$ with $v\ne w$, $f(v)\precneqq  f(w)$ implies $I(v)<I(w)$.
\end{proof}

\medskip

Given a vertex $v$ of $\sS$ and a simplex $\sigma\in \sS$ with vertices affinely independent on $v$, we denote by $v*\sigma$ the {\em join} of $v$
and $\sigma$ which is, in our geometric setting, the convex hull of $\{v\}\cup \sigma$. We further denote by ${\sS}'(v)$ the lower link of
$v$ which is defined by the following formula
\begin{equation}\label{lower-link}
{\sS}'(v)=\{\tau\in {\sS}\mid v*\tau\in {\sS} \wedge \, \mbox{$\forall$ vertex $w\in \tau$,  $f(w)\precneqq f(v)$}\}.
\end{equation}

\begin{alg}{\em [Matching]
\label{alg:matching}

\noindent{\bf Input:} A finite simplicial complex $\sS$ with a function $f : {\sS}_0 \to \R^k$ and an indexing $I:{\sS}_0 \to \N$ on its vertices.

\noindent{\bf Output:} Three lists $\sA,\sB,\sC$ of simplices of $\sS$, and a function $\ma:\sA\to \sB$.\\

\noindent
\textbf{function} \textsf{Partition} (complex $\sS$, function $f$, indexing map $I$)\\
\noindent{\bf Begin}
\begin{enumerate}
\item Initially, set $\sA, \sB, \sC=\emptyset$.
\item For each $v\in {\sS}_0$,
\begin{enumerate}
\item  Compute ${\sS}'(v)$, the lower link of $v$.
\item If ${\sS}'(v)$ is empty, then add $v$ to $\sC$.  Else
\begin{enumerate}
\item add $v$ to $\sA$.
\item let $f':{\sS}'_0(v)\to \R^k$ be the restriction of $f$ and $I':{\sS}'_0(v)\to \N$ be the restriction of $I$.
\item Call \textsf{Partition} (recursively) with input arguments ${\sS}'(v)$, $f'$, and $I'$, and get the output $\sA', \sB', \sC',\ma'$.
\item Set $\sD'=\{w\in \sC_0' \mid \mbox{$f(w)$ is minimal in $\sC_0'$ w.r.t. $\precneqq$}\}$.
\item Set $w_0$ as the vertex with smallest index $I'$ in $D'$.
\item Add $[v,w_0]$ to $\sB$ and define $\ma(v)=[v,w_0]$.
\item For each $\sigma \in \sC'\setminus \{w_0\}$, add $v*\sigma$ to $\sC$.
\item For each $\sigma \in \sA'$, add $v*\sigma$ to $\sA$, add $v*\ma'(\sigma)$ to $\sB$, and define $\ma(v*\sigma)=v*\ma'(\sigma)$.
\end{enumerate}
\end{enumerate}
\item endfor.
\item For each $\sigma\in {\sS}\setminus (\sA\cup \sB\cup \sC)$, add $\sigma$ to $\sC$.
\item \textbf{return} $\sA$, $\sB$, $\sC$, $\ma$.
\end{enumerate}
\noindent{\bf End}
}
\end{alg}

\begin{lem}\label{lem:partition}
$\sA, \sB, \sC$ is a partition of ${\sS}$ and $\ma$ is a bijective function from $\sA$ to $\sB$.
\end{lem}

\begin{proof} $\sA \cup \sB \cup \sC= {\sS}$ by instruction 4. By construction (instruction viii), the map $\ma$ is onto.
We show, by induction on the dimension of simplices in $\sA$ and $\sB$, that $\sA \cap \sB=\emptyset$ and that $\ma$ is injective.
The proof of the equalities $\sA \cap \sC = \emptyset = \sB \cap \sC$ goes by similar arguments and we leave it to the reader.
By instructions (b) and (i) vertices cannot belong to $\sB$. Therefore the first claim is true for simplices of dimension 0.
Moreover, the function $\ma$ restricted to vertices of $\sA$ is necessarily bijective. Indeed, if the edge $[v,w_0]$ is assigned
to $v \in \sA$ (instruction vi), that is $\ma(v)=[v,w_0]$, it cannot be assigned again to $w_0$, because this would require that
$v \in \sS'(w_0)$. Then $f(v) \precneqq f(w_0)$ and $f(w_0) \precneqq f(v)$ implying $v=w_0$, a contradiction.

Let us now assume that the claim is true for simplices of dimension less than $n$. Let $\tau$ be a simplex of dimension $n$ in
$\sA \cap \sB$. By instruction (viii) and since $\tau \in \sA$, there exists $\sigma \in \sA'$ (where $\sA'=\sA'(v)$) such that
$\tau=v*\sigma$ and $\ma(\tau) = v*\ma' (\sigma)$. Since $\tau \in \sB$, there exits $\beta \in \sA$ such that $\tau = \ma (\beta)$.
Since $\dim \beta > 0$, there must exist a vertex $w$ and a simplex $\gamma \in \sA'$ such that $\beta = w * \gamma$ and
$\tau = w * \ma'(\gamma)$, with $\ma'(\gamma) \in \sB'$. The vertices $v$ and $w$ must be equal, otherwise they must belong
to the lower link of each other which would be a contradiction. It follows that $\sigma = \ma' (\gamma) \in \sA' \cap \sB'$
(where  $\sB'=\sB'(v)$) which violates the induction hypothesis.

We have proved that $\sA \cap \sB=\emptyset$ and we pass to the injectivity of $\ma$. Let $\tau_1, \tau_2 \in \sA$ be simplices
of dimension $n$ such that $\ma (\tau_1) = \ma (\tau_2)$. There must exist vertices $v_1, v_2$ and simplices
$\sigma_1, \sigma_2 \in \sA'$ such that $\tau_1 = v_1 * \sigma_1$, $\tau_2 = v_2 * \sigma_2$ and
\[
\ma (\tau_1) = v_1 * \ma'(\sigma_1) = \ma (\tau_2) =  v_2 * \ma'(\sigma_2).
\]
From what precedes, we can see that the vertices $v_1, v_2$ must be equal, otherwise they must belong to the lower link of each
other. It follows that $\ma'(\sigma_1) = \ma'(\sigma_2)$ and, by the induction hypothesis, we must have $\sigma_1 = \sigma_2$
and therefore $\tau_1 = \tau_2$, which completes the proof.
\end{proof}

We define the map $\mI: \sS \to \R$ on simplices as follows
\[
\mI (\sigma) = \max_{v \, \, \mbox{\footnotesize vertex of}\, \, \sigma} I(v).
\]

\begin{lem}\label{lem:nondecreasing}
\begin{enumerate}
\item[]
\item[(a)] For every $\sigma < \tau$, $\mI (\sigma) \leq \mI (\tau)$.
\item[(b)] For every $\sigma \in \sA$, $\mI (\sigma) = \mI (\ma (\sigma)).$
\end{enumerate}
\end{lem}
\begin{proof}
(a) is trivial from the definition of $\mI$.
(b) If $\sigma$ is a vertex $v$, then $\ma (v) =[v,w]$ for some $w \in \sS'(v)$.
By Lemma~\ref{lem:index}, $I(w) < I(v)$ and hence $\mI(v) = \mI (\ma (v))$.
Let $\sigma \in \sA$ be a simplex of dimension $n \geq 1$. There exist a vertex $v$ and
a simplex $\sigma' \in \sA' \subset \sS'(v)$ such that
$\sigma = v* \sigma'$ and $\ma (\sigma) = v * \ma' (\sigma')$. Since $\sigma'$ and $\ma' (\sigma')$ are simplices of
the lower link of $v$, it follows that both $\mI (\sigma')$ and $\mI (\ma'(\sigma'))$ are smaller than  $I (v)$. Thus
$$\mI (\sigma) = \mI(v*\sigma')=\mI(v*\ma(\sigma'))=\mI (\ma(\sigma)) = I(v).$$
\end{proof}

\begin{thm}\label{th:alg}
Algorithm~\ref{alg:matching} produces a partial matching $(\sA,\sB,\sC,\ma)$ that is acyclic.
Moreover, if $\sigma\in \sS^\alpha$ then $\ma(\sigma)\in \sS^\alpha$.
\end{thm}

\begin{proof}
The partial matching is acyclic if and only if it is a gradient vector field of a discrete Morse function.
From~\cite[Theorem 6.2]{Forman02}, this is equivalent to prove that there are no nontrivial closed directed paths in the
modified Hasse diagram of the complex $\sS$. Assume that
\begin{equation}\label{eq:directed-loop}
\ell: \sigma_0 \xrightarrow{\ma}  \tau_0 \xrightarrow{\succ} \sigma_1 \xrightarrow{\ma} \tau_1 \ldots \xrightarrow{\ma}
\tau_n \xrightarrow{\succ} \sigma_0
\end{equation}
is a directed loop in the modified Hasse diagram, where $\ma$ stands for the matching and the symbol $\succ$ for the
face relation. From Lemma~\ref{lem:nondecreasing},
we deduce that $\mI$ is nondecreasing along any directed path in the modified Hasse diagram.
It follows that $\mI$ has to be constant along any directed loop.
Thus, there must exist a unique vertex $v$ such that
\[
\mI (\sigma_0) = \mI (\tau_0) = \mI (\sigma_1) = \ldots = \mI (\tau_n) = I(v),
\]
and $v$ must belong to every $\sigma_i, \tau_i \in \ell$. We will prove by induction that this leads to a contradiction.
First, observe that if $\dim \sigma_0 = \dim \sigma_1 = \ldots = 0$, then either these vertices are equal, in which case the loop is trivial,
or they are distinct in which case $\mI = I$ (on vertices) cannot be constant since it is injective. It follows that we cannot have a directed
loop $\ell$ with cells of dimensions $0$ and $1$. Assume this claim is true up to dimensions $n-2$ and $n-1$, and suppose our directed loop
$\ell$ in~(\ref{eq:directed-loop}) is composed of cells $\sigma_i$ of dimension $n-1$ and cells $\tau_i$ of dimension $n$. We have proved that
$v$ is a vertex of each $\sigma_i, \tau_i \in \ell$, so there exist simplices $\sigma_0^\prime$, $\sigma_1^\prime$, \ldots, $\sigma_n^\prime$,
$\tau_0^\prime$, $\tau_1^\prime$, \ldots, $\tau_n^\prime$ in $\sS$ such that
$\sigma_i = v* \sigma_i^\prime$ and $\tau_i = v* \tau_i^\prime$. It is easily seen that $\sigma_{i+1} \prec \tau_{i}$ implies
that $\sigma_{i+1}^\prime \prec \tau^\prime_{i}$. On the other hand, $\tau_{i} = \ma (\sigma_i)$ means that there must exist a vertex
$w_i$ and a simplex $\gamma_i^\prime \in A' \subset \sS'(w_i)$ such that $\sigma_i = w_i * \gamma_i^\prime$ and $\tau_i = w_i * \ma'(\gamma_i^\prime)$.
Using the same arguments as in the proof of Lemma~\ref{lem:partition}, we conclude that we must have $v = w_i$ and therefore
$\gamma_i^\prime = \sigma_i^\prime$ and $\ma'(\gamma_i^\prime) = \tau_i^\prime$. This shows also that $\sigma_i^\prime$ and
$\tau_i^\prime$ have to be in $\sS'(v)$. We can see now that we have a directed loop
\[
\ell': \sigma_0^\prime \xrightarrow{\ma'}  \tau_0^\prime \xrightarrow{\succ} \sigma_1^\prime \xrightarrow{\ma'} \ldots \xrightarrow{\ma'}
\tau_n^\prime \xrightarrow{\succ} \sigma_0^\prime
\]
in the modified Hasse diagram of $\sS'(v)$ with simplices of dimensions $n-2$ and $n-1$, which violates the induction
hypothesis.

Let now $\sigma$ be a simplex of $(\sS, \kappa)$ such that $\sigma \in \sS^\alpha$.
By definition of $\ma$, there exist a vertex $v$ and simplices $\sigma^\prime, \tau^\prime \in \sS^\prime (v)$ such that
$\sigma = v* \sigma^\prime$ and $\ma (\sigma) = v* \tau^\prime$. By definition of lower link and $\sS^\alpha$, it follows that for every vertex
$w$ in $\sigma^\prime$ or $\tau^\prime$, $f(w) \precneqq f(v) \preceq \alpha$. Hence, $\ma (\sigma) \in \sS^\alpha$.
\end{proof}

\begin{rem}\label{rem-weak-ls}{\em
A variation of partial matching may be obtained by replacing the lower link ${\sS}'(v)$ in formula~(\ref{lower-link}) with the {\em weak lower link} defined by
\[
{\sS}''(v)=\{\tau\in {\sS}\mid v*\tau\in {\sS} \wedge \, \mbox{$\forall$ vertex $w\in \tau$,  $f(w)\preceq f(v)$}\},
\]
and analogously replacing $\precneqq$ by $\preceq$ in the definition of $\sD'$. The condition $v*\tau\in {\sS}$ implies that $v$ is not in  its weak lower link.
The injectivity of $I$ and the instruction 2(b)-v of the algorithm permit carrying on the proofs. We considered this version of the algorithm with the hope of
matching more cells, however our experiments did not show a significant improvement in terms of getting a more accurate set $\sC$.
}\end{rem}

\subsection{Complexity Analysis}\label{complexity}
We first describe the computational complexity of Algortihm~\ref{alg:matching}.
Let $d$ be the dimension of the complex $\sS$. For each $\sigma \in \sS$, we define $\deg (\sigma)$ to be the cardinality
of the set of all cofaces of $\sigma$.

We recall that $N$ is defined to be the cardinality of $\sS_0$, i.e. the number of vertices in $\sS$. For a vertex $v \in \sS_0$,
its lower link  $\sS'(v)$, which is a subcomplex of $\sS$, consists of at most $\deg (v)$ simplices of dimensions smaller or equal
to $d-1$. It follows that $\sS'(v)$ has at most $\deg (v) d$ vertices. If we assume the worst case scenario where every
vertex has a nonempty lower link, the first call to function Partition will result in $N$ subsequent calls for
Partition, each for a fixed vertex $v \in \sS_0$, with arguments $\sS'(v)$ and the restrictions
of $f$ and $I$ to $\sS'(v)$. Since $\deg (v)$ varies for each vertex $v$, it is difficult to establish any complexity bounds for the algorithm
without assuming some constraints on $\deg (v)$. We will assume hereafter that $\deg (v)$ is bounded above by a constant $\gamma$ for every
$v \in \sS_0$. This is a reasonable assumption when dealing with complexes of manifolds and approximating surface boundaries of objects.
For each vertex $v \in \sS_0$, we need to examine its set of cofaces (read directly from the structure storing the complex)
to create its lower link which can be done first in at most $\gamma$ steps. The partition of the subcomplex $\sS'(v)$ (resulting from
the recursive call to Partition) will be visited once (in at most $\gamma$ steps) to execute the steps (b)-vi to (b)-viii of the algorithm.
We assume that the vertices of $\sS$ are already ordered with respect to the indexing function. It is easily seen that any subsequent call to
\textsf{Partition} with a complex formed by a lower link of some vertex and of dimension $s <d$ is completed in a number of operations
directly proportional to the number of simplices and vertices in the complex which are bounded by $\gamma$ and $\gamma (s+1)$ respectively.
This number will be denoted by $\alpha (\gamma, s)$.

\begin{thm}\label{th:complexity}
Algorithm~\ref{alg:matching} produces a partial matching $(\sA,\sB,\sC,\ma)$ in less than
$2 \gamma^{d} (d+1)! N$ steps.
\end{thm}
\begin{proof}
From the discussion above, we deduce that the processing of each vertex of $\sS$ is completed in less than
$2\gamma + \alpha (\gamma, d-1)$ operations. Therefore, the number of operations for processing all the vertices (call it $\eta$) is bounded above by
$$N(2\gamma + \alpha (\gamma, d-1)).$$
Reasoning by induction and using the arguments above, each subsequent call to \textsf{Partition} on a complex of dimension $s$ of
a lower link of a vertex costs less than $\alpha (\gamma, s)$ and we have
$$\alpha (\gamma, s) \leq \gamma (s+1) \left(2\gamma + \alpha (\gamma, s-1)\right).$$
Moreover, when the complex consists only of vertices, each of them will have an empty lower link in that complex. Thus, we can
conclude that $\alpha (\gamma,0) \leq \gamma$. Putting all together, we can conclude now that
\begin{eqnarray}
\nonumber
\eta & \leq & N \left[2\gamma + \alpha (\gamma, d-1)\right] \leq N \left[2\gamma + \gamma d \left[ 2 \gamma + \alpha (\gamma, d-2)\right]\right]\\
\nonumber
& \leq & N \left[2\gamma + 2 \gamma^2 d + \gamma d \alpha (\gamma, d-2)\right].
\end{eqnarray}
By induction, we can prove that
\begin{eqnarray}
\nonumber
\eta & \leq & N \left[2\gamma + 2 \gamma^2 d + \ldots + 2\gamma^{(d-1)} d(d-1) \ldots 2 + \gamma^{(d-1)} d(d-1) \ldots 2 \alpha (\gamma, 0)\right]\\
\nonumber
& \leq & 2 \gamma^{d} (d+1)! N.
\end{eqnarray}

\end{proof}

Let $n$ denote the total number of cells in the original complex $\sS$. The computation of the rank invariant
of a $d$-dimensional multi-filtration of the complex $\sS$ may be achieved with an algorithm that runs in
$O(n^{2d+3})$ operations (see \cite{CaSiZo09} for more details). Our method which consists of using the acyclic
matchings to perform homology preserving reductions on the original complex, allows to postpone the persistent
homology computation until the complex is reduced to a smaller one which may yield a tremendous gain in the
number of operations incurred. Let $m$ denote the number of cells in the final complex after all reductions yielded by the
acyclic matching are performed. Thus, the computational cost of the multidimensional persistent homology of the
complex $\sS$ is reduced to $O(m^{2d+3})$. To illustrate the significance of our method, let us assume that
our matching algorithm allows to reduce the complex by half its number of cells (a ratio that is comparable to the ones
provided in our experimental results). In this case, the persistent homology computational cost is reduced by a factor
of $2^{2d+3}$ (slightly greater than $500$ if $d=3$) when the computation is performed on the reduced complex. This is a major gain
when comparing the computationally inexpensive reduction with the time consuming persistent homology computation.

Indeed, if we assume that we work under the constraint that $\deg (\sigma) \leq \gamma$ for every
$\sigma \in \sS$, we can easily prove that each elementary reduction is achieved in constant time. Hence, the time complexity
of the total reduction process which runs through all the matching pairs $\{\ma(\sigma),\sigma\}$ and performs the reductions
is in the worst case linear in the number of cells of the complex.

Our aim is to make $m$ very small compared to $n$, or equivalently construct an optimal acyclic matching. However, this problem is
known to be NP hard~\cite{JosPfet06} and there are no known procedure to minimize $m$ for arbitrary complexes. In our context where 
we are dealing with a multidimensional function and using an algorithm based on exploring lower links of vertices, it is possible 
to reach an outcome where no reduction is possible and every cell is critical. This point is illustrated in Figure~\ref{cancellation}(a).

Since our work is inspired by the work in~\cite{KinKnuMra05}, it is natural to raise the question of whether it is
possible to add some cancelling step to reduce further the number of critical cells and allow a bigger number of reductions before
proceeding with the persistent homology computation. Since our algorithm produces an acyclic matching of the complex, it is possible to build 
gradient paths and do cancellations when possible as defined in~\cite{For98}. However, the cancellation of critical cells is not necessarily 
desirable in this context because it works against providing a full account of the history of births and deaths of homology generators which is necessary
for obtaining complete information about the persistent homology. This latter point can be easily illustrated with a simple example
as shown in Figure~\ref{cancellation}(b).

\begin{figure}
\begin{center}
\begin{tabular}{cc}
\includegraphics[height=3.5cm]{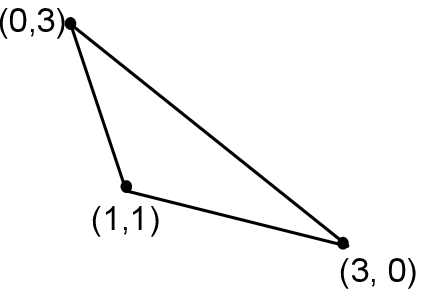}&
\includegraphics[height=3.5cm]{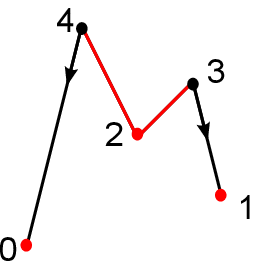}\\
(a) & (b)
\end{tabular}
\caption{(a) Example of a one dimensional complex with a 2-dimensional map on its vertices and in which every cell is critical according to Algorithm~\ref{alg:matching}. 
(b) Example of a one dimensional complex with a one dimensional map on its vertices that can
be extended to a Discrete Morse Function by assigning to each one dimensional cell the maximum of the values of its vertices.
We can easily see that any cancellation of cells would lead to a change in persistence homology of the complex.
}
\label{cancellation}
\end{center}
\end{figure}

\subsection{Back to Reductions}\label{sec:main-reductions}

In this section, we prove that an acyclic matching on an $\sS$-complex $(\sS, \kappa)$ allows by means of reductions
to replace the initial complex by a smaller one with the same persistent homology. The motivation for this approach
stems from the need to achieve a low computational cost in the persistence homology computation.

In the sequel, we assume that $(\sA,\sB,\sC,\ma)$ is an acyclic matching on a filtered $\sS$-complex $\sS$ with the property:

\begin{equation}\label{matching-filtration}
\mbox{If } \sigma\in\sS^\alpha \mbox{ then } \ma(\sigma)\in \sS^\alpha.
\end{equation}

Theorem~\ref{th:alg} asserts that the matching produced by Algorithm~\ref{alg:matching} on a filtered simplicial complex $\sS$ has this property.

\begin{prop}\label{prop:chain-maps-include}
Let $\sigma\in \sA$ and let $(\overline{\sS}, \overline{\kappa})$ be obtained from $(\sS,\kappa)$ by reduction of the pair $(\ma(\sigma),\sigma)$.
Let $\pi$, $\iota$, and $D$ be maps defined by formulas (\ref{eq:pi}), (\ref{eq:iota}), and (\ref{eq:chain-homotopy}) respectively.
Then $\pi(C_*(\sS^\alpha))\subseteq C_*(\overline{\sS}^\alpha)$, $\iota(C_*(\overline{\sS}^\alpha))\subseteq C_*(\sS^\alpha)$,
and $D_q(C_q(\sS^\alpha))\subseteq C_{q+1}(\overline{\sS}^\alpha)$, for each $q\in \Z$.
\end{prop}

\begin{proof}
Let $\tau\in \sS^\alpha$. We need to show that $\pi(\tau)\in C_*(\overline{\sS}^\alpha)$. By definition of $\pi$, the only non trivial case is
when $\tau=\sigma$. In this case, $\sigma\in \sS^\alpha$ and by (\ref{matching-filtration}),  $\ma(\sigma)\in \sS^\alpha$.  Note that the chain
$\pi(\sigma)$ is supported in the union of cells $\xi\in\overline{\sS}$ such that $\kappa(\ma(\sigma),\xi)\ne 0$. Each such $\xi$ is a face of
$\ma(\sigma)\in \sS^\alpha$, hence $\xi\in \overline{\sS}^\alpha$.

Let now $\tau \in \overline{\sS}^\alpha$. We need to show that $\iota(\tau)\in C_*(\sS^\alpha)$. By definition of $\iota$,
the only non trivial case is when $\kappa(\tau,\sigma)\ne 0$.  This implies that $\sigma$ is a face of $\tau$. Let
$\tau\in \overline{\sS}^\alpha$. By Definition~\ref{def:filtration-reduced}, this means that $\tau\in \sS^\alpha$.
By definition of filtration, it follows that $\sigma\in \sS^\alpha$. Again, by (\ref{matching-filtration}), $\ma(\sigma)\in \sS^\alpha$,
proving the claim.

The statement on $D_*$ instantly follows by the same argument.
\end{proof}

\begin{lem}\label{lem:chain-equivalence}
The maps $\pi_{|C_*(\sS^\alpha)}:C_*(\sS^\alpha)\to C_*(\overline{\sS}^\alpha)$ and
$\iota_{|C_*(\overline{\sS}^\alpha)}:C_*(\overline{\sS}^\alpha)\to C_*(\sS^\alpha)$ defined by restriction are chain homotopy equivalences.
Moreover, the diagram
\[
\begin{array}{ccc}
H_*(\sS^\alpha) & \mapright{H_*(j^{(\alpha,\beta)})} & H_*(\sS^\beta) \\
\mapdown{\cong} & & \mapdown{\cong} \\
H_*(\overline{\sS}^\alpha) & \mapright{H_*(j^{(\alpha,\beta)})} & H_*(\overline{\sS}^\beta)
\end{array}
\]
commutes.
\end{lem}

\begin{proof}
By Proposition~\ref{prop:chain-maps-include}, we have the commutative diagram
\[
\begin{array}{ccc}
C_*(\sS^\alpha) & \hookrightarrow & C_*(\sS^\beta) \\
\mapdown{\pi_{|C_*(\sS^\alpha)}} & & \mapdown{\pi_{|C_*(\sS^\beta)}} \\
C_*(\overline{\sS}^\alpha) & \hookrightarrow & C_*(\overline{\sS}^\beta)
\end{array}
\]
where the vertical arrows are chain equivalences. The result follows by the functoriality of homology.
\end{proof}

This lemma immediately yields the following result.

\begin{thm}\label{th:reduced-filtration-iso}
For every $\alpha\preceq \beta\in \R^k$, $H_*^{\alpha,\beta}(\sS)$ is isomorphic to $H_*^{\alpha,\beta}(\overline{\sS})$.
\end{thm}

Let us order $\sA$ in a sequence
\[
\sA=\{\sA(1),\sA(2),\ldots,\sA(n)\}
\]
and set $\sB(i)=\ma(\sA(i))$, $i=1,2,\ldots,n$. Put $\sS(0)=\sS$ and
\[
\sS(i)=\overline{\sS(i-1)}=\sS(i-1)\setminus \{\sB(i),\sA(i)\},\;\;i=1,2,\ldots,n.
\]
Since a partial matching defines a partition of $\sS$, we have $\sS(n)=\sC$.

Note that, by Definition~\ref{def:filtration-reduced}, the condition (\ref{matching-filtration}) carries through to the reduced complex.
Consequently, Corollary~\ref{cor:reduced-matching}, Lemma~\ref{lem:chain-equivalence} and Theorem~\ref{th:reduced-filtration-iso} extend
by induction to any step of reduction. Hence, for any $\alpha\in \R^k$,  we get a sequence of filtered $\sS$-complexes
\[
(\sS^{\alpha}(0),\kappa^{\alpha}(0)),\; (\sS^{\alpha}(1),\kappa^{\alpha}(1)),\;\ldots,\;(\sS^{\alpha}(n),\kappa^{\alpha}(n)),
\]
where $\kappa^{\alpha}(i)=\overline{\kappa^{\alpha}(i-1)}$, together with a sequence of chain equivalences
\[
\pi^{\alpha}(i): C_*(\sS^{\alpha}(i-1))\to C_*(\sS^{\alpha}(i)),\;\;\iota^{\alpha}(i): C_*(\sS^{\alpha}(i))\to C_*(\sS^{\alpha}(i-1)).
\]
Moreover, for any $\alpha\preceq \beta$, we get the sequence of inclusions
\[
j^{(\alpha,\beta)}(i):\sS^{\alpha}(i)\hookrightarrow \sS^{\beta}(i),
\]
such that the commutative diagram of Lemma~\ref{lem:chain-equivalence} applied to the $i$'th iterate gives the following.
\[
\begin{array}{ccc}
H_*(\sS^\alpha(i-1)) & \mapright{H_*(j^{(\alpha,\beta)}(i-1))} & H_*(\sS^\beta(i-1)) \\
\mapdown{\cong} & & \mapdown{\cong} \\
H_*(\sS^\alpha(i)) & \mapright{H_*(j^{(\alpha,\beta)}(i))} & H_*(\sS^\beta(i))
\end{array}
\]
By induction, we get the the following.
\begin{cor}\label{cor:homology-S-iso-C}
For every $\alpha\preceq \beta\in \R^k$, $H_*^{\alpha,\beta}(\sS)$ is isomorphic to $H_*^{\alpha,\beta}(\sC)$.
\end{cor}

\section{Experimental Results and Conclusion}\label{sec:experiments}

We considered four triangle meshes (available at \cite{sf}).
Each mesh was filtered by  the 2-dimensional measuring function $f$
 taking each vertex $v$ of coordinates $(x, y, z)$ to the pair $f(v) = (|x|, |y|)$.

In Table \ref{tab:space}, the first row shows on the top line the number of vertices in each considered mesh, and in
the middle line same quantities referred to the cell complex $\sC$
obtained by using our matching algorithm to reduce $\sS$. Finally, it also displays in the bottom line the ratio between
the second and the first lines, expressing them in percentage points.
The second and the third rows show similar information for the edges and the faces. Finally, the fourth row show the same
information for the total number of cells of each considered mesh $\sS$.

Our experiments confirm that the current vertex-based matching algorithms do not produce optimal reduction of the complex
so that every remaining cell is relevant in the computation of persistence homology. The discussion and the examples 
provided in subsection~\ref{complexity} show the limitations of this method. 
They show a fair rate of reduction for vertices, but the reduction rate for cells of dimensions 1 and 2 is not as
significant as that for vertices.

\begin{table}
\caption{Reduction performance on some triangle meshes.}
\begin{center}
\begin{tabular}{| c | c | c | c | c |}
\hline
 & $\raisebox{-.9\height}{\includegraphics[height=1.5cm]{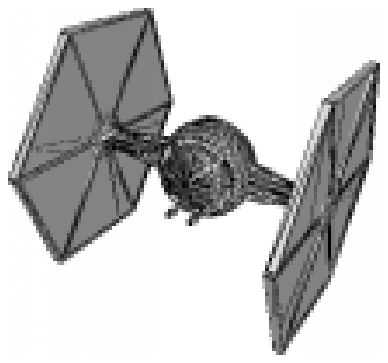}}$& $\raisebox{-.9\height}
 {\includegraphics[height=1.5cm]{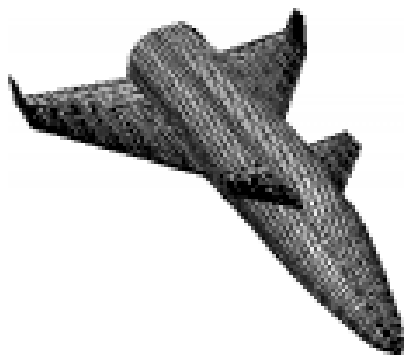} }$&  $\raisebox{-.9\height}
 {\includegraphics[height=1.5cm]{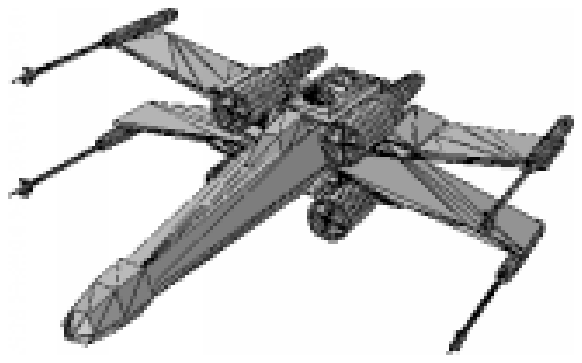}}$&  $\raisebox{-.9\height}
 {\includegraphics[height=1.5cm]{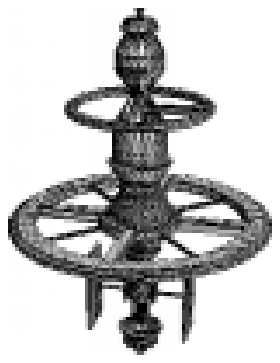}}$\\
 & \tt{tie} & \tt{space\_shuttle} & \tt{x\_wing} & \tt{space\_station}\\
 \hline
$\begin{array}{rrr}\# \sS_0\\ \# \sC_0\\ \% \end{array} $&
$\begin{array}{rrr}2014\\228\\ 11.3 \end{array}$&
$\begin{array}{rrr}2376\\121\\ 5.1 \end{array}$ &
$\begin{array}{rrr}3099\\175\\ 5.6 \end{array}$&
$\begin{array}{rrr}5749\\1879\\ 32.7 \end{array}$ \\
 \hline
 $\begin{array}{rrr}\# \sS_1\\ \# \sC_1\\ \% \end{array}$ &
 $\begin{array}{rrr}5944\\3343\\ 56.2 \end{array}$&
 $\begin{array}{rrr}6330\\3699\\ 58.4 \end{array}$ &
 $\begin{array}{rrr}9190\\3605\\ 39.2 \end{array}$ &
 $\begin{array}{rrr}15949\\11158\\ 70.0 \end{array}$ \\
 \hline
 $\begin{array}{rrr}\# \sS_2\\ \# \sC_2\\ \% \end{array}$ &
 $\begin{array}{rrr}3827\\3012\\ 78.7 \end{array}$ &
 $\begin{array}{rrr}3952\\3576\\ 90.5 \end{array}$&
 $\begin{array}{rrr}6076\\3415\\ 56.2 \end{array}$ &
 $\begin{array}{rrr}10237\\9316\\ 91.0 \end{array}$ \\
 \hline
 $\begin{array}{rrr}\# \sS\\ \# \sC\\ \% \end{array}$ &
 $\begin{array}{rrr}11785\\6583 \\55.9  \end{array}$ &
 $\begin{array}{rrr}12658\\7396 \\58.4  \end{array}$ &
 $\begin{array}{rrr}18365\\7195 \\39.2  \end{array}$ &
 $\begin{array}{rrr}31935\\22353 \\70.0  \end{array}$\\
 \hline
 \end{tabular}
 \end{center}
\label{tab:space}
\end{table}

 \begin{center}%
 {\bfseries Acknowledgments\vspace{-.5em}}%
 \end{center}%
This work was partially supported by the following institutions: INdAM-GNSAGA (C.L.), NSERC Canada Discovery Grant (T.K.).

\bibliographystyle{abbrv}
\bibliography{biblio}

\medskip

\noindent Department of Computer Science\\
Bishop's University\\
Lennoxville (Qu\'ebec),  Canada J1M 1Z7\\
mallili@ubishops.ca
\\~\\
D\'epartement de math\'ematiques\\
Universit\'e de Sherbrooke,\\
Sherbrooke (Qu\'ebec), Canada J1K 2R1\\
t.kaczynski@usherbrooke.ca
\\~\\
\noindent Dipartimento di Scienze e Metodi dell'Ingegneria\\
Universit\`a di Modena e Reggio Emilia\\
Reggio Emilia, Italy
\\
claudia.landi@unimore.it
\end{document}